\documentclass[final,authoryear,3p,times,twocolumn]{elsarticle}
\usepackage{amssymb}
\usepackage{amsmath}
\usepackage{amsthm} 
\newtheorem{theorem}{Theorem}
\journal{Insurance: Mathematics and Economics}

\begin{document}

\begin{frontmatter}

\title{Optimal Retirement Consumption with a Stochastic Force of Mortality\tnoteref{comments}}
\tnotetext[comments]{This research is supported in part by grants from NSERC, MITACS/Mprime and The IFID Centre. 
The authors acknowledge helpful comments from an IME
reviewer, as well D. Babbel, T. Davidoff, G. Harrison, M. Lachance, S.F. Leung and from
seminar participants at the Lifecycle, Insurance, Finance and Economics (LIFE)
conference at the Fields Institute, Georgia State University, McMaster
University, the Risk Theory Seminar (RTS) at Florida State University and the University of Technology, Sydney. 
A previous version of this paper circulated as: \emph{Yaari's lifecycle model in the 21st century: 
consumption under a stochastic force of mortality}}
\author[math]{Huaxiong Huang}
\ead{hhuang@yorku.ca}
\author[schulich]{Moshe A. Milevsky\corref{corresp}}
\ead{milevsky@yorku.ca}
\author[math]{Thomas S. Salisbury}
\ead{salt@yorku.ca}
%Huaxiong Huang is Professor of Mathematics and Statistics at York University
%Moshe A. Milevsky is Associate Professor of Finance, York University, and Executive Director of the IFID Centre
%Thomas S. Salisbury is Professor of Mathematics and Statistics at York University
%Version: 6 April 2012
\address[math]{Dept. of Mathematics \& Statistics, York University, Toronto}
\address[schulich]{Schulich School of Business, York University, Toronto}
\cortext[corresp]{Corresponding author}

\begin{abstract}

We extend the lifecycle model (LCM) of consumption over a random horizon 
(a.k.a. the Yaari model) to a world in which (i.) the
force of mortality obeys a diffusion process as opposed to being
deterministic, and (ii.) a consumer can adapt their consumption strategy to
new information about their mortality rate (a.k.a. health status) as it
becomes available. In particular, we derive the optimal consumption rate and focus on
the impact of mortality \emph{rate} uncertainty vs. simple \emph{lifetime}
uncertainty -- assuming the actuarial survival curves are initially identical -- in the
retirement phase where this risk plays a greater role.

In addition to deriving and numerically solving the PDE for the optimal
consumption rate, our main general result is that when utility
preferences are logarithmic the initial consumption rates are identical. But,
in a CRRA framework in which the coefficient of relative risk aversion is
greater (smaller) than one, the consumption rate is higher (lower) and a
stochastic force of mortality does make a difference. 

That said, numerical experiments indicate that even for non-logarithmic preferences, 
the stochastic mortality effect is relatively minor from the individual's perspective.
Our results should be relevant to researchers interested in calibrating the
lifecycle model as well as those who provide normative guidance (a.k.a.
financial advice) to retirees.

\end{abstract}

\end{frontmatter}

\section{Introduction and Motivation}
\label{intro}

The lifecycle model (LCM) of savings and consumption -- originally postulated
by \citet{F1930} and refined by \citet{MB1954}, \citet{M1986} -- is at
the core of most multi period asset pricing and allocation models, as well as the
foundation of microeconomic consumer behavior. The original formulation -- for
example \citet{R1928} and \citet{P1962} -- assumed a deterministic horizon.
But, in a seminal contribution, the LCM was extended by \citet{Y1964,Y1965} to
a stochastic lifetime, which eventually led to the models of \citet{M1971},
\citet{R1975} and hundreds of subsequent papers on asset allocation over the
human lifecycle.

The conceptual under-pinning of the LCM is the intuitive notion of
\emph{consumption smoothing} whereby (rational) individuals seek to minimize
disruptions to their standard of living over their entire life. They
plan a consumption profile that is continuous, equating marginal
utility at all points, based on the assumption of a concave utility function.
See the recent (and very accessible) article by \citet{K2008} in which this
concept is explained in a non-technical way. 

Once again, until the seminal contribution by \citet{Y1964, Y1965}, the LCM was
employed by economists in an idealized world in which death occurred with
probability one at some terminal horizon. Menahem Yaari introduced lifetime
uncertainty into the lifecycle model, in addition to -- his more widely known
contribution of -- introducing actuarial notes and annuities into optimal
consumption theory.

In the expressions (and theorems) he derived for the optimal consumption
function, \citet{Y1965} assumed a very general \emph{force of mortality} for
the remaining lifetime random variable, without specifying a particular law.
His results would obviously include a constant force of mortality (i.e.
exponential remaining lifetime) as well as Gompertz-Makeham (GM) mortality,
and other commonly formulated approximations. Yaari provided a rigorous
foundation for Irving Fisher's claim that lifetime uncertainty effectively
increases consumption impatience and is akin to behavior under higher
subjective discount rates. Mathematically, the mortality rate was added to the
subjective discount rate.

That said, most of the empirical or prescriptive papers in the LCM literature
have not gone beyond assuming the GM law -- or some related deterministic
function -- for calibration purposes.  In other words, mortality is just a substitute 
for subjective discount rates. In fact, one is hard-pressed to differentiate high levels of
longevity and mortality risk aversion from weak preferences for consumption today vs. 
the future, i.e. patience. Some have labeled this risk neutrality with respect to lifetime uncertainty.

For example: \citet{LM1977},
\citet{D1981}, \citet{D1991}, \citet{L1994}, \citet{B2001}, \citet{BDOW2004}, 
\citet{DL2005}, \citet{KT2005}, \citet{BM2006}, \citet{P2006},
\citet{WZ2007}, \citet{F2008}, or the recent work by \citet{L2012}
-- all assume a deterministic force of mortality.

Indeed, some economists continue (surprisingly) to ignore mortality alltogether, for
example the recent review by \citet{AW2010}. Perhaps this is because when the force
of mortality is deterministic, it can be added to the subjective discount rate without 
any impact on the mathematical structure of the problem.

To our knowledge, the only authors within the financial economics literature that have considered
the possibility of non-constant mortality rates in a lifecycle model are \citet{CG2009}, 
although their Lee-Carter mortality model is not quite
\emph{stochastic} as in \citet{MP2001}, \citet{D2004}, \citet{CBD2006}, or the 
various models described in the the book by \citet{PDHO2008}, or the concerns 
expressed by \citet{N2010}.

Moreover, a number of very recent papers -- for example \citet{M2008},
\citet{S2009} and \citet{P2010} -- have examined the implications of (truly)
stochastic mortality rates on the demand and pricing of certainly annuity
products, but have not derived the impact of stochasticity on optimal
consumption alone or examined the impact of pure uncertainty in the mortality rate.

Another related paper is \citet{BV} who examine the impact of
relaxing the assumption of additively separable utility and what-they-call risk neutrality
with respect to life duration. But, they also assume a deterministic force of moratlity in
their formulation and examples. In that sense, our work is similar because we also relax 
the so-called risk neutrality and the intertemportal additivity.

In sum, to our knowledge, none of the existing papers within the LCM
literature have assumed a stochastic force of mortality -- which is 
the model of choice in the current actuarial and insurance literature --
and then derived its impact on pure consumption behavior. We believe this to
be a foundational question, and in this paper our objective is
straightforward, namely, to compare the impact of stochastic vs. deterministic
mortality rates on the optimal consumption rate.

\subsection{A Proper Comparison}
\label{comparison}

Assume that two hypothetical retirees -- i.e. consumers who are not expecting
any future labour income -- approach a financial economist for guidance on how
they should spend their accumulated financial capital over their remaining
lifetime; a time horizon they both acknowledge is stochastic. Assume both
retirees have time-separable and rational preferences and seek to maximize
discounted utility of lifetime consumption under the same elasticity of
intertemporal substitution ($1/\gamma$), the same subjective discount rate
($\rho$) and the same initial financial capital constraint ($F_{0}$). They
have no declared bequest motives and -- for whatever reason -- neither are
willing (or able) to invest in anything other than a risk-free asset with
instantaneous return ($r$); which means they are \emph{not} looking for
guidance on asset allocation or annuities.\footnote{This simplification is
made purely to focus attention on the impact of stochastic mortality.}
All they want is an optimal consumption plan
($c^{\ast}(t);t\geq0$) guiding them from time zero (retirement) to the last
possible time date of death $(t\leq D)$. Most importantly, both retirees
agree they share the same probability-of-survival curve denoted by $p(s)$.
In other words they currently live in the same health state\ and the same
effective biological age. For example, they both agree on a $p(35)=5\%$
probability that either of them survive for 35 years and a $p(20)=50\%$
probability that either of them survive for 20 years, etc.

\citet{Y1964,Y1965} showed exactly how to solve such a problem. He
derived the Euler-Lagrange equation for the optimal trajectory of wealth and
the related consumption function.

In Yaari's model both of the above-mentioned retirees would be
told to follow identical consumption paths until their random date of death.
In fact, they would both be guided to optimally consume $c(t)^{\ast
}=F(t)/a(t)$, where $a(t)$ is a \emph{function of time only} and is related to
an actuarial annuity factor. We will explain this factor in more detail, later
in the paper.

But here is the impetus for our comparison. Although both retirees appear to
have the same longevity risk assessment and agree-on the survival probability
curve $p(s)$, they have \emph{differing views about the volatility of their
health as proxied by a mortality rate volatility}. In the language of
current actuarial science, the first retiree (1) believes that his
instantaneous force of mortality (denoted by $\lambda^{\text{DfM}}(t)$) will
grow at a deterministic rate until he eventually dies, while the second
retiree (2) believes that her force of mortality (denoted by $\lambda
^{\text{SfM}}(t)$) will grow at stochastic (but measurable) rate until a
random date of death. As such, the remaining lifetime random variable for retiree
2 is doubly stochastic. While this distinction might sound farfetched and
artificial, a growing number of researchers in the actuarial literature are
moving to such models \footnote{We appreciate and acknowledge comments made by a 
referee, that models in which mortality depends on health status, which itself is stochastic,
have been used by actuaries well-before the introduction of 21st century stochastic mortality 
models.}, rather than the simplistic mortality models traditionally
used by economists. The actuaries' motivation in advocating for a
stochastic force of mortality, is to generate more robust pricing and
reserving for mortality-contingent claims. These studies have all argued that
SfM models better reflect the uncertainty inherent in demographic projections
\emph{vis a vis} the inability of insurance companies to diversify mortality
risk entirely. We ask: \emph{how do the recent actuarial models impact the
individual economics of the problem?}

When one thinks about it, real-life mortality rates are indeed stochastic,
capturing (unexpected) improvements in medical treatment, or (unexpected)
epidemics, or even (unexpected) changes to the health status of an individual.
Rational consumers choosing to make saving and consumption decisions using
models based on deterministic mortality rates would likely agree to
re-evaluate those decisions if their views about the values of those mortality
rates change dramatically. Our thesis is that economic decision-making can
only be improved if mortality models reflect the realistic evolution of mortality rates.

We will carefully explain the mathematical distinction between deterministic
and stochastic forces of mortality (SfM) in Section \ref{forceofmortality} of this paper, but
just to make clear here, at time zero both our hypothetical retirees agree on
the initial survival probability curve $p(s)$. However, at any future time
their perceived survival probability curves will deviate from each other
depending on the realization of the mortality rate between now and then.

Motivated by such models of mortality, in this paper we
derive the optimal consumption function for both retirees; one who believes in
-- and operates under a -- stochastic mortality and one who does not. Stated
differently, we will solve the (consumption only) \citet{Y1965} model where the
optimal consumption plan is given as a function of wealth, time \emph{and the
evolving mortality rate as a state variable}. Indeed, with thousands of LCM
papers in the economic literature over the last 50 years, and the growing interest in stochastic
mortality models in the actuarial community, we believe these results will be
of interest to both communities of researchers.

Recall that in the Yaari model conditioning on the mortality rate was
redundant or unnecessary since its evolution over time was deterministic. All
one needed was the value of wealth $F(t)$ and time $t$. But, in a stochastic
mortality model, the mortality rate itself becomes a state variable. In this
paper we show how the uncertainty of mortality interacts with longevity risk
aversion ($\gamma$) -- which is the reciprocal of the intertemporal elasticity
of substitution -- to yield an optimal consumption plan. Mortality no longer functions 
as just a discount rate.

To briefly preview our results, we describe the conditions under which retiree
1 (deterministic mortality) will start-off consuming more than retiree 2
(stochastic mortality), as well the conditions under which retiree 1
consumes less than retiree 2, and the (surprising) conditions under which
they both consume exactly the same. We provide numerical examples under a
variety of specific mortality models and examine the magnitude of this effect.

The remainder of this paper is organized as follows. In Section \ref{forceofmortality} 
we explain
in more detail exactly how a stochastic model of mortality differs from the
more traditional (and widely used in economics) deterministic force of
mortality. In Section \ref{reviewYaari} we take the opportunity to review the 
(consumption only)
\citet{Y1965} model and set our notation and benchmark for the stochastic
model. In Section \ref{optimalConsumption} we characterize the 
optimal consumption function in the
stochastic mortality model under the most general assumptions, and prove a
theorem regarding the relationship between consumption in the two models. In
Section \ref{examples} we make some specific assumptions 
regarding the stochastic
mortality rate and illustrate the magnitude of this effect, and Section \ref{conclusion}
summarizes our main results and concludes the paper. The appendix contains
mathematical details and algorithms that are not central to our main economic contributions.

First, we explain exactly the difference between deterministic and stochastic
force of mortality.

\section{Understanding The Force of Mortality}
\label{forceofmortality}

In most of the relevant papers in the LCM literature over the last 45 years
the force of mortality from time zero to the last possible date of death is
known with certainty. Ergo, the conditional survival probabilities over the
entire retirement horizon are known (in advance) at time zero. So, if a
65-year-old retiree is told (by his doctor) that he faces a 5\% chance of
surviving to age 100 and a 37\% chance of surviving to age 90, then by
definition there is a 13.5\% = (0.05/0.37) probability of surviving to age
100, if he is still alive at age 90. In other words, he makes consumption
decisions today that trade-off utility in different states of nature, knowing
that if-and-when he reaches the age of 90, there will only be a 13.5\% chance
he will survive to age 100. In the language of actuarial science, the table of
individual $\{q_{x+i};i=0,\dots,N\}$ mortality rates is known in advance. This
is the essence of a deterministic force of mortality and textbook life
contingencies. If $q_{65}$ is the retiree's probability of dying between age
65 and 66, while $q_{66}$ is the probability of the same retiree dying between
age 66 and 67, then the probability of surviving from age 65 to age 67 is
$(1-q_{65})(1-q_{66})$.

In stark contrast, under a stochastic force of mortality the above
multiplicative relationship breaks down. We do not know in advance how 
survival probabilities will evolve. While a
65-year-old might currently face a 5\% estimated probability of surviving to
age 100 and a 37\% chance of reaching age 90, there is absolutely no guarantee
that the conditional survival probability from any future age, to age 100
(given the observed mortality rates), will satisfy the ratio. At time zero
there is an expectation of what the probability will be at age 90. But, the
probability itself is random. This way of thinking -- which might be new to
economists -- is the essence of a stochastic force of mortality and is the
impetus for our paper.

Here it is formally. Let $\lambda(t)$ denote the mortality rate of a cohort of
a population, which may be stochastic or deterministic. Let $\mathcal{F}%
_{t}=\sigma\{\lambda(q)\mid q\leq t\}$ be the filtration determined by
$\lambda$. Then individuals in the population have lifetimes of length
$\zeta$ satisfying
\begin{equation}
P(\zeta>s\mid\zeta>t,\mathcal{F}_{\infty})=e^{-\int_{t}^{s}\lambda(q)\,dq}.
\end{equation}
Assume further that $\lambda(t)$ is a Markov process, and define the survival
function $p(t,s,\lambda)$ by
\begin{equation}
p(t,s,\lambda)=E\left[  e^{-\int_{t}^{s}\lambda(q)\,dq}\mid\lambda
(t)=\lambda\right]  .
\end{equation}
This gives the conditional probability of surviving from time $t$ to time $s$,
given knowledge of the mortality rate at time $t$. Therefore
\begin{multline}
P(\zeta>s\mid\zeta>t,\mathcal{F}_{t})\\
=E\left[  e^{-\int_{t}^{s}\lambda
(q)\,dq}\mid\mathcal{F}_{t}\right]  =p(t,s,\lambda(t)).
\end{multline}
If $t=0$ then we write $p(s,\lambda)$ for $p(0,s,\lambda)$.

Our basic problem in this paper will be to compare optimal consumption under
two models that share a common initial value $\lambda_{0}$ of the mortality
rate, as well as a common survival function $p(t,\lambda_{0})$. Typically one
will be deterministic and one stochastic. When we do actual computations, we
will either choose a specific deterministic model and calibrate a stochastic
model to it, or conversely, we will choose a stochastic model and calibrate
the deterministic model to it. Both possibilities are discussed below. It
should be clear from the context which model we are discussing. But when it is
necessary to make this distinction explicitly, we will write $\lambda
^{\text{DfM}}(t)$ and $\lambda^{\text{SfM}}(t)$.

\subsection{Deterministic force of Mortality (DfM)}
\label{dfm}

Let $\lambda_{0}=\lambda(0)$ be the initial value of the mortality rate. In
the deterministic case,
\begin{equation}
p(t,\lambda_{0})=e^{-\int_{0}^{t}\lambda(q)\,dq},
\end{equation}
and we can recover $\lambda(t)$ as $-p_{t}(t,\lambda_{0})/p(t,\lambda_{0})$,
where the $t$-subscript denotes the time derivative. In other words, if we
start with a concrete stochastic model, and obtain the survival curve
$p(t,\lambda_{0})$ from it, the above formula determines the calibration of
the deterministic force of mortality model. This approach is
computationally simpler, but has the disadvantage that neither the stochastic
nor deterministic model is in a simple form, familiar to and used by practitioners. In
other words, a \textquotedblleft simple\textquotedblright\ model for the
stochastic force of mortality rates leads to a \textquotedblleft
complicated\textquotedblright\ model for the deterministic force of mortality,
and vice versa.

When doing actual calculations we will start by assuming that $\lambda(t)$
follows a standard Gompertz model. 
The Gompertz model was introduced in 1825, but more 
recently was popularized by \citet{C1994}, for example. Alternative models
are presented in \citet{GV1998} and others are discussed as
early as \citet{B1961}. In our case, we use:
\begin{equation}
d\lambda(t)=\eta\lambda(t)\,dt
\end{equation}
so $\lambda(t)=\lambda_{0}e^{\eta t}$. The usual form for Gompertz is
$\lambda(t)=b^{-1}e^{(x+t-m)/b}$, so here we are using $\eta=1/b$ and
$\lambda_{0}=b^{-1}e^{(x-m)/b}$. This model is simple, and takes advantage of
long experience calibrating the Gompertz model to real populations.

Note that in the deterministic setting,
\begin{multline}
p(t,s,\lambda(t))=e^{-\int_{t}^{s}\lambda(q)\,dq}\\
=e^{-\int_{0}^{s}%
\lambda(q)\,dq}/e^{-\int_{0}^{t}\lambda(q)\,dq}=\frac{p(s,\lambda_{0})}{p(t,\lambda
_{0})}. \label{deterministic.rule}%
\end{multline}
This will typically NOT be true in the stochastic setting. As long as we keep
in mind that we are calibrating at time 0 (i.e. to $p(t,\lambda_{0})$ only)
that should not cause problems.%

\begin{figure*}
\begin{center}
\begin{tabular}
[c]{||c||c||c||c||c||c||c||c||c||}\hline\hline
\multicolumn{9}{||c||}{\textbf{Table 1: Conditional Survival Probability:
Deterministic Mortality}}\\\hline\hline
& $x=65$ & $x=70$ & $x=75$ & $x=80$ & $x=85$ & $x=90$ & $x=95$ &
$x=100$\\\hline\hline
To Age 65 & 1.000 &  &  &  &  &  &  & \\\hline\hline
To Age 70 & 0.9479 & 1.000 &  &  &  &  &  & \\\hline\hline
To Age 75 & 0.8659 & 0.9135 & 1.000 &  &  &  &  & \\\hline\hline
To Age 80 & 0.7429 & 0.7837 & 0.8580 & 1.000 &  &  &  & \\\hline\hline
To Age 85 & 0.5733 & 0.6047 & 0.6620 & 0.7716 & 1.000 &  &  & \\\hline\hline
To Age 90 & 0.3696 & 0.3899 & 0.4268 & 0.4975 & 0.6447 & 1.000 &  &
\\\hline\hline
To Age 95 & 0.1758 & 0.1855 & 0.2031 & 0.2367 & 0.3067 & 0.4757 & 1.000 &
\\\hline\hline
To Age 100 & 0.0500 & 0.0527 & 0.0577 & 0.0673 & 0.0872 & 0.1353 & 0.2844 &
1.000\\\hline\hline
$\lambda(x)$ & 0.0081 & 0.0137 & 0.0232 & 0.0394 & 0.0667 & 0.1129 & 0.1911 &
0.3234\\\hline\hline
\end{tabular}
\end{center}
\end{figure*}

Table 1 displays a typical (loosely based on U.S. unisex annuitant mortality)
deterministic mortality survival probability \textquotedblleft
matrix\textquotedblright\ of values together with the corresponding mortality
rate at each age $x$, on the bottom row. Note that these numbers were
generated using a (deterministic) Gompertz model in which $m=89.335$ and
$b=9.5$. Indeed, given the initial probability of survival from age 65 to any
age $y>65$ (which is the first column in Table 1) one can solve for the
conditional survival probability from age $y$ to any age $z>y$, by dividing
the two probability values. This is the essence of equation
(\ref{deterministic.rule}). Alas, when mortality rates are stochastic all
numbers $p(t,s,\lambda(t))$ beyond the first column in Table 1, are unknown
at time zero.

\subsection{Stochastic Force of Mortality (SfM)}
\label{sfm}

There are many possible stochastic models to choose from. Starting from the
models of \citet{LC1992}, \citet{CBD2006}
as well as \citet{WS2010}, actuaries have employed a variety of
specifications for the stochastic $\lambda(t)$, subsequently used to price
mortality and longevity risk. In what follows in the numerical examples, we
adopt a lognormal mortality rate, which is often called the Dothan model for
interest rates in the derivative pricing literature -- see \citet{D1978}. Although it might seem
natural to have constant drift and diffusion coefficients, in order to
calibrate to a given deterministic model, we allow a time-dependent growth
coefficient. For most of the numerical examples provided later-on we take:%
\begin{equation}
d\lambda(t)=\mu(t)\lambda(t)\,dt+\sigma\lambda(t)\,dB(t) \label{SFM.model1}%
\end{equation}
where $B(t)$ is a Brownian motion. This is obviously the source of randomness
in the stochastic force of mortality.
%Our simpler numerical experiments (in Section \ref{examples}) will be for the special
%case of GBM, in which $\mu(t)=\mu$ is constant. In that case, we can obtain a
%backward equation for $p(t,\lambda_{0})$. In particular, we fix $T$ and use
%time-homogeneity to obtain an expression
%\begin{equation}
%E\left[  e^{-\int_{0}^{T}\lambda(q)\,dq}\mid\mathcal{F}_{t}\right]
%=e^{-\int_{0}^{t}\lambda(q)\,dq}p(T-t,\lambda(t)).
%\end{equation}
%This must be a martingale, so applying It\^o's lemma, we get that $p(t,\lambda)$
%satisfies the partial differential equation (PDE):
%\begin{equation}
%-p_{t}+\mu\lambda p_{\lambda}+\frac{\sigma^{2}\lambda^{2}}{2}p_{\lambda
%\lambda}-\lambda p=0.\label{survival.PDE}%
%\end{equation}
%Equation (\ref{survival.PDE}) is an important intermediary step in all our
%calculations since it provides us with the survival probability as a function
%of the initial hazard rate (or age), the time horizon and drift and
%diffusion coefficients $\mu,\sigma$ driving the hazard rate. We will return
%to calibrating the time-inhomogeneous case in the appendix, which is not
%necessary for considering optimal consumption in general. For now we simply
%illustrate one example.
There are many ways to select (or calibrate) a stochastic force of mortality
to a particular survival curve.
%We then show the
%probability that $p(t,s,\lambda(t))>\frac{p(s,\lambda_{0})}{p(t,\lambda_{0}%
%)}$. That is, the conditional survival probability for the stochastic model
%exceeds that for the deterministic model.
The details on how to actually compute this are provided in the second part of
the appendix.

With the probability background out of the way, we now review the (consumption
only) \citet{Y1965} model which is based on a deterministic force of mortality.

\section{Review of the \citet{Y1965} Model}
\label{reviewYaari}

The canonical lifecycle model (LCM) with a random date of death and assuming no
bequest motive, can be written as follows:
\begin{equation}
J=\max_{c}E\Big[  \int_{0}^{D}e^{-\rho t}u(c(t))1_{\{t\leq\zeta\}}dt\Big]
,\label{objective}%
\end{equation}
where $\zeta$ is the remaining lifetime satisfying $\Pr
[\zeta>t]=p(t,\lambda_{0})$, defined above in Section \ref{forceofmortality}. 
We fix a (deterministic) 
last possible time $D$ of death, so $\zeta\le D$. When the mortality
rate is deterministic one can obviously assume independence between the
optimal consumption $c^{\ast}(t)$ and the lifetime indicator variable
$1_{\{t\leq\zeta\}},$ so that by Fubini's theorem we can re-write the value
function as:%
\begin{align}
J &  =\max_{c}\int_{0}^{D}e^{-\rho t}u(c(t))E[1_{\{t\leq\zeta\}}]dt\nonumber\\
&  =\max_{c}\int_{0}^{D}e^{-\rho t}u(c(t))p(t,\lambda_{0})dt.\label{objective.2}
\end{align}
From this perspective, there really is not any more randomness in the model.
This is a problem within the calculus of variations subject to some
constraints on the function $c(t)$. In the end, the survival probability is absorbed into the discount rate.

Let $r$ denote the risk free interest rate. To avoid the distractions of inflation 
models and assumptions, throughout this
paper we assume that $r$ is expressed in real
(after-inflation) terms and therefore consumption $c(t)$ is expressed in real
terms as well. The wealth (budget) constraint can then be
written as:%
\begin{equation}
F_{t}(t)=rF(t)+\pi_{0}-c(t),\label{budget}%
\end{equation}
with boundary conditions $F(0)=W>0$ and $F(D)=0$. We are using the subscript
$F_{t}$ to denote a first derivative w.r.t time, and if needed $F_{tt}$ for
the second derivative. The parameter $\pi_{0}$ denotes a constant income rate
which we include in this section for comparison with Yaari's model, but which
in subsequent sections will be taken to equal zero; $c(t)$ is the consumption
rate and the control variable in our problem.
%more generally, we can let $v=v(t,F)$ be the interest rate
%function defined by:
%\begin{equation}
%v(t,F)=\left\{
%\begin{array}
%[c]{cc}%
%r, & F\geq0,\\
%R, & F<0,
%\end{array}
%\right.
%\end{equation}
%where $r\leq R\,\leq\infty$. In words, $v:=v(t,F)$ is the investment return
%$r$ when wealth is positive so that $F\geq0.$ It is equal to the borrowing
%rate $R$ when wealth is negative, ie. $F<0$. \ Note that the above-formulation
%is precisely Case B of \citet{Y1965}, and when $R=\infty$ we are effectively
%imposing a (no) borrowing constraint. 
In a follow-up paper we hope to examine the impact of additional factors, 
such as different interest  rates for borrowing versus lending, 
of the availability of actuarial notes (i.e. the case when the
interest rate is $r+\lambda(t)$).

In this paper we operate under a constant relative risk aversion (CRRA)
formulation for the utility function. In principle this should mean using
$\bar{u}(c)$, where:
\begin{equation}
\bar{u}(c)=\frac{c^{1-\gamma}-1}{1-\gamma}%
\end{equation}
for $\gamma>0$ and $\gamma\neq1$, with the understanding that when $\gamma=1$
we define $\bar{u}(c)=\ln c$. This family of utilities varies continuously
with $\gamma$. The marginal utility of consumption is the derivative of utility with respect
to $c$, which is simply
\begin{equation}
u_{c}=c^{-\gamma}>0.
\end{equation}

Of course, it makes no difference to our optimization problem (and the optimal control) 
if we shift $\bar{u}$ by an arbitrary additive constant. So to make scaling relationships
easier to express, actual calculations will be carried out using the
equivalent utilities
\begin{equation}
u(c)=\frac{c^{1-\gamma}}{1-\gamma}%
\end{equation}
for $\gamma>0$ and $\gamma\neq1$. When $\gamma=1$ we take $u(c)=\bar{u}%
(c)=\ln c$.

As a consequence of the Euler-Lagrange Theorem, the optimal financial capital
trajectory $F(t)$ must satisfy the following linear second-order
non-homogenous differential equation over the values for which $F(t)\neq0.$
\begin{multline}
F_{tt}(t)-\left(  \frac{r-\rho-\lambda(t)}{\gamma}+r\right)  F_{t}(t)+r\left(
\frac{r-\rho-\lambda(t)}{\gamma}\right)  F(t)\\
=-\left(  \frac{r-\rho
-\lambda(t)}{\gamma}\right)  \pi_{0}. \label{Basic.ODE}%
\end{multline}
When the pension income rate $\pi_{0}=0$ the differential equation collapses
to the homogenous case. See \citet{C1992}, for examplem for an exposition of Euler-Lagrange
equations in economics. 

\subsection{Explicit Solution: Gompertz Mortality}
\label{Gompertz}

When the (deterministic) mortality rate function obeys the (pure) Gompertz law
of mortality
\begin{equation}
\lambda(t)=\frac{1}{b}\exp\left(  \frac{x+t-m}{b}\right)  ,
\end{equation}
the survival probability can be expressed as
\begin{equation}
p(t,\lambda_{0})=e^{-\int_{0}^{t}\lambda(q)\,dq}
=e^{  b\lambda_{0}(1-e^{t/b})}  .
\end{equation}
Here $x$ denotes the age at time 0, $m$ is called the modal value and $b$ is
the dispersion coefficient for the Gompertz model. To simplify notation let
\begin{equation}
k(t)=\frac{r-\rho-\lambda(t)}{\gamma}, \label{k.define}%
\end{equation}
and recall from the budget constraint that:
\begin{align}
c(t)  &  =rF(t)-F_{t}(t)+\pi_{0},\label{c.eq1}\\
c_{t}(t)  &  =rF_{t}(t)-F_{tt}(t). \label{c.eq2}%
\end{align}
Equation (\ref{Basic.ODE}) can be rearranged as
\begin{equation}
F_{tt}(t)-rF_{t}(t)+k(t)(rF(t)-F_{t}(t))=-k(t)\pi_{0}, \label{Basic.ODE2}%
\end{equation}
which then leads to%
\begin{equation}
k(t)c^{\ast}(t)-c_{t}^{\ast}(t)=0 \label{ODE.C}%
\end{equation}
The solution to this basic equation is%
\begin{multline}
c^{\ast}(t)    =c^{\ast}(0)e^{\int_{0}^{t}k(s)ds}=c^{\ast}(0)e^{\int_{0}%
^{t}\frac{r-\rho-\lambda(s)}{\gamma}  ds}\\
=c^{\ast}(0)e^{
\frac{r-\rho}{\gamma} t}e^{-\frac{1}{\gamma}\int_{0}^{t}%
\lambda(s)\,ds}\label{consumption}
=c^{\ast}(0)e^{\frac{r-\rho}{\gamma} t}p(t,\lambda
_{0})^{1/\gamma},
\end{multline}
where $c^{\ast}(0)$ is the optimal initial consumption rate, to be determined,
which is the one free constant resulting from equation (\ref{ODE.C}). Note
that when the interest rate $r$ is equal to the subjective discount rate
$\rho$, and $\gamma=1$ (i.e. log utility), the optimal consumption rate at any
age $x+t$ is the probability of survival to that age times the initial
consumption $c^{\ast}(0)$. However, when $\gamma>1$, which implies higher
levels of risk aversion, the optimal consumption rate will decline at a slower
rate as the retiree ages. Longevity risk aversion induces people to behave as
if they were going to live longer than determined by the actuarial mortality
rates. We will explore the impact of $\gamma$ on the optimal consumption path
in a stochastic force of mortality model, later in 
Section \ref{optimalConsumption}, which is why
it's important to focus on this here.

Mathematically one can see that $(p(t,\lambda_{0}))^{1/(\gamma+\varepsilon)}$
is greater than $(p(t,\lambda_{0}))^{1/\gamma}$ for any $\varepsilon>0$ since
$p(t,\lambda_{0})<1$ for all $t$. Finally, note that in the Gompertz mortality
model evaluating $(p(t,\lambda_{0}))^{1/\gamma}$ for a given $(x,m,b)$ triplet
is equivalent to evaluating $p(t,\lambda_{0})$ under the same $x,b$ values,
but assuming that $m^{\ast}=m+b\ln\gamma$. This then implies that one can
tilt/define a new deterministic mortality rate $\hat{\lambda}_{0}%
=\gamma\lambda_{0}$ and derive the optimal consumption as if the individual
was risk neutral. This will be used later in the explicit expression for
$F(t)$ and $c^{\ast}(t)$

Moving on to a solution for $F(t)$, we now substitute the optimal consumption
solution (\ref{consumption}) into equation (\ref{c.eq1}) to arrive at yet
another first-order ODE, but this time for $F(t)$:%
\begin{equation}
F_{t}(t)-rF(t)-\pi_{0}+c^{\ast}(0)e^{\frac{r-\rho}{\gamma}
t}p(t,\lambda_{0})^{1/\gamma}=0. \label{ODE.W}%
\end{equation}
Writing down the canonical solution to this equation leads to:
\begin{multline}
F(t)=  e^{rt}\left(  \pi_{0}\int_{0}^{t}e^{-rs}ds\right.\\
\left.-c^{\ast}(0)\int_{0}^{t}e^{
\frac{r-\rho}{\gamma}  s}p(s,\lambda_{0})^{1/\gamma}e^{-rs}%
ds+F(0)\right), \label{wealth}%
\end{multline}
where $F_{0}$ denotes the free initial condition from the ODE for $F(t)$ in
equation (\ref{ODE.W}). Recall that we still haven't specified $c^{\ast}(0)$,
the initial consumption). We will do so (eventually) by using the terminal
condition $F(D)=0$.

To represent the wealth trajectory explicitly define the following (new)
Gompertz Present Value (GPV) function%
\begin{align}
a_{x}^{T}(r,&m,b)   =\int_{0}^{T}p(s,\lambda_{0})e^{-rs}ds=\int_{0}%
^{T}e^{-\int_{0}^{s}\left(  r+\lambda(t)\right)  dt}ds\nonumber\\
&  =\int_{0}^{T}e^{-\int_{0}^{s}\left(  r+\frac{1}{b}e^{\left(  \frac
{x-m+t}{b}\right)  }\right)  dt}ds\label{GPV}\\
&  =\frac{b\Gamma(-rb,\exp\{\frac{x-m}{b}\})-b\Gamma(-rb,\exp\{\frac{x-m+T}%
{b}\})}{\exp\{(m-x)r-\exp\{\frac{x-m}{b}\}\}}.\nonumber%
\end{align}
The function $a_{x}^{T}(r,m,b)=a(t)$ is the age$-x$ cost of a life-contingent
annuity that pays \$1 per year continuously provided the annuitant is still
alive, but only until time $t=T$, which corresponds to age $x+T$. If the
individual survives beyond age $(x+T)$ the payout stops. Naturally, when
$T=\infty$ the expression collapses to a conventional single premium income
annuity (SPIA).

Note that $\Gamma(A,B)$ is the incomplete Gamma function. In other words,
equation (\ref{GPV}) is analytic and in closed-form.

The reason for introducing the GPV is that combining equation (\ref{wealth})
with equation (\ref{GPV}) leads to the (very tame looking) expression%
\begin{equation}
F(t)=\left(  F(0)+\frac{\pi}{r}\right)  e^{rt}-a_{x}^{t}(r-k,m^{\ast
},b)c^{\ast}(0)e^{rt}-\frac{\pi_{0}}{r}, \label{wealth.explicit}%
\end{equation}
where recall that $m^{\ast}=m+b\ln\gamma$. Then, using the boundary
condition $F_{\tau}=0$, where $\tau\,$\ is the wealth depletion time, we
obtain an explicit expression for the initial consumption%
\begin{equation}
c^{\ast}(0)=\frac{\left(  F(0)+\pi_{0}/r\right)  e^{r\tau}-\pi_{0}/r}%
{a_{x}^{\tau}(r-k,m^{\ast},b)e^{r\tau}}. \label{optimal.consumption.DFM}%
\end{equation}

\subsection{Consumption Under DfM: Numerical Examples}
\label{consumpEGs}

In our numerical examples we assume an 86.6\% probability that a 65-year-old
will survive to the age of 75, a 57.3\% probability of reaching 85, a 36.9\%
probability of reaching 90, a 17.6\% probability of reaching age 95 and a 5\%
probability of reaching 100. These are the values generated by the Gompertz
law with $m=89.335$ and $b=9.5$. To complete the parameter specifications
required for our model, we assume the subjective discount rate ($\rho$) is
equal to the risk-free rate $r=2.5\%$ Within the context of a lifecycle model,
this implies that the optimal consumption rates would be constant over time in
the absence of longevity and mortality uncertainty.

We are now ready for some results. Assume a 65-year-old with a (standardized)
\$100 nest egg. Initially we allow for no pension annuity income ($\pi_{0}=0$)
and therefore all consumption must be sourced to the investment portfolio
which is earning a deterministic interest rate $r=2.5\%$. The financial
capital $F(t)$ must be depleted at the very end of the lifecycle, which is
time $D=(120-65)=55$ and there are no bequest motives. So, according to
equation (\ref{optimal.consumption.DFM}), the optimal consumption rate at
retirement age 65 is \$4.605 when the risk aversion parameter is $\gamma=4$
and the optimal consumption rate is (higher) \$4.121 when the risk aversion
parameter is set to (higher) $\gamma=8$.

As the retiree ages $(t>0)$ he/she rationally consumes less each year -- in
proportion to the survival probability adjusted for $\gamma$. For example, in
our baseline $\gamma=4$ level of risk aversion, the optimal consumption rate
drops from \$4.605 at age 65, to \$4.544 at age 70 (which is $t=5$), then
\$4.442 at age 75 (which is $t=10$), then \$3.591 at age 90 (which is $t=25$)
and \$2.177 at age 100 (which is $t=35$), assuming the retiree is still alive.
A lower real interest rate ($r$) leads to a reduced optimal
consumption/spending rate. All of this can be sourced to equation
(\ref{consumption}).

Thus, one of the important insights is that a fully rational consumer will
actually spend less as they progress through retirement. The optimizer spends
more at earlier ages and reduces spending with age, even if his/her subjective
discount rate (SDR) is equal to (or less than) the real interest rate in the economy.

Intuitively the individual deals with longevity risk by planning to reduce
consumption -- if that risk materializes -- in proportion to the survival
probability, linked to their risk aversion. The \citet{Y1965} model provides a
rigorous foundation to the statement by \citet{F1930} in his book \emph{Theory
of Interest} (page 85): \emph{\textquotedblleft\ldots The shortness of life
thus tends powerfully to increase the degree of impatience or rate of time
preference beyond what it would otherwise be\ldots\textquotedblright\ }and
(page 90)\emph{ \textquotedblleft Everyone at some time in his life doubtless
changes his degree of impatience for income\ldots When he gets a little
older,\ldots he expects to die and he thinks: instead of piling up for the
remote future, why shouldn't I enjoy myself during the few years that
remain?\textquotedblright}\footnote{For additional (case specific) examples of
the \citet{Y1965} model in action during the non-labour income retirement
phase, we refer the interested reader to \citet{MH2010} or a recent
paper by \citet{L2012}.}

\subsection{Time-zero Consumption Ratio = Initial Withdrawal Rate}
\label{time0consump}

Finally, in the very specific case when $\pi_{0}=0$ (which implies that the
wealth depletion time is $\tau=D$) and the subjective discount rate $\rho=r$,
the retiree must rely exhaustively on his/her initial wealth $F_{0}$. We get%
\begin{equation}
\frac{c^{\ast}(0)}{F(0)}=\frac{1}{a_{x}^{D}(r-\lambda_{0}/\gamma,m^{\ast},b)}%
\end{equation}
We now have all the ingredients to compare with a stochastic model. This ratio
is often called the Initial Withdrawal Rate (IWR) amongst financial
practitioners and in the retirement spending literature.

\section{Optimal Consumption: General Results}
\label{optimalConsumption}

In this section we obtain the \emph{most general} optimal consumption strategy
for a retiree maximizing expected discounted utility of consumption with uncertain lifetime, 
which will include the (consumption only) \citet{Y1965} model as a special case.
Since our main focus now is on the mortality model, at this stage we make the
additional assumption $\rho=r$, that is, that the subjective discount rate
equals the interest rate in the economy. Also, in contrast to the discussion
in the previous section, we assume no exogenous pension income, so that
$\pi_{0}=0$, which then precludes any borrowing. We now assume a fixed
terminal horizon $T$, which denotes the last possible date of death. The
mathematical formulation is to find
\begin{multline}
J=\max_{\text{$c(s)$ adapted}}E\left[  \int_{0}^{T}e^{-\int_{0}%
^{s}(r+\lambda(q))\,dq}u(c(s))ds\right.\\
\left.\vphantom{\int_0^T}
\Big\vert \,\lambda(0)=\lambda,F(0)=F\right]
\end{multline}
Whereas in Section \ref{reviewYaari} of this paper we used calculus of variation 
techniques to
derive the optimal trajectory of wealth and the consumption function, 
given the inclusion of mortality as a state variable we
must resort to dynamic programing techniques to obtain the optimality
conditions. Regardless of the different techniques, we will show how the
optimal consumption function collapses to the \citet{Y1965} model when the
volatility of mortality is zero.

Define:
\begin{multline}
J(t,\lambda,F)=\max_{\text{$c(s)$ adapted}}E\left[  \int_{t}^{T}
e^{-\int_{t}^{s}(r+\lambda(q))\,dq}u(c(s))ds\right.\\
\left.\Big\vert\, \lambda(t)\vphantom{\int_t^T}
=\lambda,F(t)=F\right]  . \label{value_function}%
\end{multline}
As in the deterministic mortality model, the wealth process (which we shall
soon see is stochastic) satisfies $dF(t)=(rF(t)-c(t))dt$. Assume that there is
an optimal control. Then for that control,
\begin{multline}
E\left[   \int_{0}^{T}e^{-\int_{t}^{s}(r+\lambda(q))\,dq}%
u(c(s))ds\,\Big\vert\, \mathcal{F}_{t}\right] \\
=e^{-\int_{0}^{t}(r+\lambda(q))\,dq}J(t,\lambda(t),F(t))\\
+\int_{0}^{t}%
e^{-\int_{t}^{s}(r+\lambda(q))\,dq}u(c(s))\,ds
\end{multline}
is a martingale. This will likewise give a supermartingale under a general
choice of $c$. Applying It\^o's lemma, we obtain the following
Hamilton-Jacobi-Bellman (HJB) equation:
\begin{multline}
\sup_{c}\left\{  u(c)-cJ_{F}\right\}  +J_{t}-(r+\lambda)J\\
+rFJ_{F}+\mu(t)\lambda J_{\lambda}+\frac{\sigma^{2}
\lambda^{2}}{2}J_{\lambda\lambda}=0.
\end{multline}
If there is any possibility of confusion, we will denote this value function
$J^{\text{SfM}}(t,\lambda,F)$.

For deterministic mortality, HJB can be obtained by sending $\sigma
\rightarrow0$ with $\mu(t)=\eta$, which was equal to $1/b$ in the \citet{Y1965}
model derived in Section \ref{reviewYaari}, as
\begin{equation}
\sup_{c}\left\{  u(c)-cJ_{F}\right\}  +J_{t}-(r+\lambda)J+rFJ_{F}+\eta\lambda
J_{\lambda}=0.
\end{equation}
%Though since $\lambda$ is no longer a state variable in this case, we could
%equally well write $J^{\text{DfM}}(t,F)=\lim_{\sigma\rightarrow0}%
%J(t,\lambda(t),F)$ and have an equation in one fewer variable:
%\begin{equation}
%\sup_{c}\left\{  u(c)-cJ_{F}^{\text{DfM}}\right\}  +J_{t}^{\text{DfM}%
%}-(r+\lambda)J^{\text{DfM}}+rFJ_{F}^{\text{DfM}}=0.
%\end{equation}
%Here the $J_{t}^{\text{DfM}}$ term accounts for both the $J_{t}$ and
%$J_{\lambda}$ terms that appeared before.

Moving on to the optimal consumption plan, we solve the HJB equation under
CRRA utility as follows: let
\begin{equation}
u(c)=\frac{c^{1-\gamma}}{1-\gamma},\quad J=\frac{F^{1-\gamma}}{1-\gamma
}a(t,\lambda),
\end{equation}
where the second expression results from the scaling which follows from the first, 
and apply the first order condition $c^{\ast}=J_{F}^{-\frac{1}{\gamma}}$. We
obtain $c^{\ast}=Fa^{-\frac{1}{\gamma}}$ and get the following equation for
$a(t,\lambda)$:
\begin{equation}
a_{t}-(r\gamma+\lambda)a+\gamma a^{1-\frac{1}{\gamma}}+\mu(t)\lambda
a_{\lambda}+\frac{\sigma^{2}\lambda^{2}}{2}a_{\lambda\lambda}=0
\label{main.PDE.1}%
\end{equation}
with boundary condition $a(T,\lambda)=0$.

We now solve the PDE for $a(t,\lambda)$, which we re-write as:
\begin{equation}
\beta_{t}+1-\Big(r+\frac{\lambda}{\gamma}\Big)\beta+\mu(t)\lambda\beta_{\lambda}+\frac{\gamma-1}{2\beta}%
\sigma^{2}\lambda^{2}\beta_{\lambda}^{2}+\frac{1}{2}\sigma^{2}\lambda^{2}%
\beta_{\lambda\lambda}=0. \label{main.PDE.2}%
\end{equation}
for $\beta=\beta(t,\lambda)=a(t,\lambda)^{1/\gamma}$ . 
The boundary conditions are
$\beta(T,\lambda)=0$, $\beta_{\lambda}(t,\infty)=0$ and at $\lambda=0$ we
solve 
$\beta_{t}+1-r \beta=0$. 
Note that the optimal consumption rate is
$c=F/\beta$, using shorthand notation. On to the main theorem.

\subsection{Stochastic Force of Mortality: Main Theorem}
\label{SfMTheorem}

Denote by $c^{\text{SfM}}(t,\lambda,F)$ the optimal consumption at time $t$,
given $\lambda(t)=\lambda$ and $F(t)=F$, under a stochastic force of mortality
(SfM) model. Denote by $c^{\text{DfM}}(t,F)$ the optimal consumption at time
$t$, when $F(t)=F$, under a deterministic force of mortality (DfM) model.

\begin{theorem} 
\label{maintheorem}
Assume that the survival functions for the
two models agree: $p^{\text{SfM}}(t,\lambda_{0})=p^{\text{DfM}}
(t,\lambda_{0})$ for every $t\geq0$, and that utility is $\text{CRRA}(\gamma)$
\begin{enumerate}
\item$\gamma>1\Longrightarrow c^{\text{SfM}}(0,\lambda_{0},F)
\geq c^{\text{DfM}}(0,F)$;
\item $\gamma=1\Longrightarrow c^{\text{SfM}}(0,\lambda_{0},F)
=c^{\text{DfM}}(0,F)$;
\item $0<\gamma<1\Longrightarrow c^{\text{SfM}}(0,\lambda_{0},F)
\leq c^{\text{DfM}}(0,F)$.
\end{enumerate}
\end{theorem}

\begin{proof}To see this, we change point of view, and work
exclusively with the stochastic model. So we drop the SfM superscript, and
write $p=p^{\text{SfM}}$, $J=J^{\text{SfM}}$, $c^{\ast}=c^{\text{SfM}}$,
$\lambda=\lambda^{\text{SfM}}$, etc. Within that model, we pose two different
optimization problems, depending on the level of information available about
$\lambda(t)$. The value function $J(t,\lambda,F)$ solves the problem given
before in (\ref{value_function}), where $c(t)$ can be any suitable process
adapted to $\mathcal{F}_{t}$. But we define a new value function $J^{1}(t,F)$
in which we impose an additional constraint on $c(t)$, namely that it be
deterministic. More precisely,
\begin{equation}
J(0,\lambda_{0},F_{0})    =\max_{\text{$c(s)$ adapted}}E\left[  \int_{0}%
^{T}e^{-\int_{0}^{s}(r+\lambda(q))\,dq}u(c(s))ds\right]
\end{equation}
and
\begin{align}
J^{1}(0,&F_{0})    \\
&=\max_{\text{$c(s)$ deterministic}}E\left[  \int_{0}%
^{T}e^{-\int_{0}^{s}(r+\lambda(q))\,dq}u(c(s))ds\right] \nonumber\\
  &=\max_{\text{$c(s)$ deterministic}}\int_{0}^{T}e^{-rs}p(s,\lambda
_{0})u(c(s))ds.\nonumber
\end{align}
We let $c^{\ast}$ denote the optimal control for $J$, and $c^{1}$ denote the
optimal control for $J^{1}$.

Since every deterministic control $c(t)$ is also adapted, we have the basic
relationship
\begin{equation}
J(0,\lambda_{0},F_{0})\geq J^{1}(0,F_{0}).
\end{equation}
On the other hand, the above expression is exactly what the old deterministic
model would have given. That is,
\begin{equation}
J^{1}(0,F_{0})=J^{\text{DfM}}(0,F_{0})
\end{equation}
and $c^{1}=c^{\text{DfM}}$.

Due to scaling, $J(t,\lambda,F)=a(t,\lambda
)F^{1-\gamma}/(1-\gamma)$ and $c^{*}(t,\lambda,F)=a(t,\lambda)^{-1/\gamma}F$
for some function $a\ge0$. Likewise $J^{\text{DfM}}(t,F)= a_{1}(t)F^{1-\gamma
}/(1-\gamma)$ and $c^{1}= a_{1}^{-1/\gamma}F$ for some $a_{1}\ge0$. If
$\gamma>1$ then $1-\gamma<0$, so $a(0,\lambda_{0})\le a_{1}(0)$, so $c^{*}\ge
c^{1}$ at $t=0$. This shows (a). Likewise if $0<\gamma<1$ then $a(0,\lambda
_{0})\ge a_{1}(0)$, so $c^{*}\le c^{1}$ at $t=0$. This shows (c).

Recall that when $\gamma=1$, we have $u(c)=\ln c$. Earlier, when $\gamma
\neq1$, we had $u(c)=c^{1-\gamma}/(1-\gamma)$ and could make use of a scaling
relation. In other words, if $c$ is optimal for $F$, then $kc$ is optimal for
$kF$, and that leads to the expression $J(t,\lambda,kF)=k^{1-\gamma
}J(t,\lambda,F)$. Or in other words,
\begin{equation}
J(t,\lambda,F)=F^{1-\gamma}J(t,\lambda,1).
\end{equation}
With logarithmic utility, the corresponding expression is that\newline%
$J(t,\lambda,kF)=J(t,\lambda,F)+(\ln k)\int_{t}^{T}e^{-r(s-t)}p(t,s,\lambda
)\,ds$. Or in other words,
\begin{equation}
J(t,\lambda,F)=J(t,\lambda,1)+(\ln F)\int_{t}^{T}e^{-r(s-t)}p(t,s,\lambda
)\,ds.
\end{equation}
Likewise,
\begin{equation}
J^{\text{DfM}}(t,F)=J^{\text{DfM}}(t,1)+(\ln F)\int_{t}^{T}e^{-r(s-t)}%
\frac{p(s,\lambda_{0})}{p(t,\lambda_{0})}\,ds.
\end{equation}
The first order conditions in the optimization problem then imply that
\begin{align}
c^{\ast}&=F/\int_{t}^{T}e^{-r(s-t)}p(t,s,\lambda)\,ds,\nonumber\\
c^{\text{DfM}}&=F/\int_{t}^{T}e^{-r(s-t)}\frac{p(s,\lambda_{0})}{p(t,\lambda_{0})}\,ds.
\end{align}
These agree when we send $t\rightarrow0$, showing (b). 
\end{proof}

The theorem certainly proves that $\gamma=1$ is a point of indifference. The
invariance of mortality volatility when utility is logarithmic is reminiscent
of similar results in consumption theory where income negates substitution
effects. More on this later.

Note that we only use $J^{1}(0,F)$ above, not $J^{1}(t,F)$. If we had, we
would have had to be careful. The correct definition is that
\begin{multline}
J^{1}(t,F)=J^{\text{DfM}}(t,F)\\
=\max_{c(s)}\int_{t}^{T}e^{-r(s-t)}%
\frac{p(s,\lambda_{0})}{p(t,\lambda_{0})}u(c(s))ds
\end{multline}
rather than
\begin{multline}
\max_{c(s)}E\left[  \int_{t}^{T}e^{-\int_{t}^{s}(r+\lambda(q))\,dq}%
u(c(s))ds\right]  \\
=\max_{c(s)}\int_{t}^{T}e^{-r(s-t)}E\Big[p(t,s,\lambda
(t))\Big]u(c(s))ds.
\end{multline}
These quantities have connections to annuities, as suggested by the fact that
the optimal consumption rates given above are, as a fraction of wealth,
inverse annuity prices. In particular, $\int_{t}^{T}e^{-rs}p(s,\lambda
_{0})\,ds$ is the (actuarial) price of a deferred annuity, purchased at time 0
with payments starting at time $t$. While $\int_{t}^{T}e^{-rs}E[p(t,s,\lambda
(t))]\,ds$ is a forward\ annuity price. That is, if at time 0 an insurance
company guarantees (a retiree) the right to buy an annuity at time $t$ at a
price determined at time 0, then this is that price (computed actuarially,
i.e. by discounting mean cash flows).

\subsection{Intuition and Relation to Known Results}
\label{intuition}

How should one interpret our result? 
It is tempting to view stochastic mortality as simply ``more risky'' that deterministic mortality, 
but that is not in fact the reason consumption shifts. The true explanation for our result is
that the comparison can be reinterpreted equivalently as one between two
different control problems, both within the context of the stochastic hazard
rate model. Namely, a control problem in which the hazard rate is observed, 
so one can react to changes, versus one in which the hazard rate is not observed,
so the control must be determined in advance. The utility in the deterministic model
is the same as the utility for the second control problem (and indeed, this is the
basis of our proof). So the mere presence of stochastic hazard rates will not
cause a change in consumption; what shifts consumption is the ability
to react to those changes. 

There are two possible reactions to that ability to adjust consumption. One is
to shift consumption into the future, taking advantage of the ability to adjust consumption
upwards later, if the hazard rate should climb more than expected. The other reaction is
to opt to consume more now, in the knowledge that one can cut back later if it seems 
likely that one will live longer than expected. Our message is that either reaction can
be rational, and that which one is adopted depends on the person's risk aversion, with
the switch occurring at the point of logarithmic utility. The choice is between acting
more conservatively in view of the possibility one might live longer, versus acting
more aggressively in view of one's ability to react to changes in the hazard rate. 

There are other results in the literature where logarithmic utility is a qualitative point 
of indifference in behavior.
An example -- in a completely different context -- is the
classical result on the equilibrium pricing of assets derived by 
\citet{L1978}\footnote{We thank Thomas Davidoff for pointing out this analogy to the
authors, and we refer the interested reader to the lecture notes by
Christopher Caroll, available at the following website:
\texttt{econ.jhu.edu/people/ccarroll/public/lecturenotes}, for this particular
interpretation (and derivation) of the Lucas model.}. In a Lucas-type model --
under logarithmic utility preferences -- the equilibrium price
of trees (or any other income producing asset for that matter) does not depend
on the projected level of fruit output from those trees. The economic reason
for this is that there are two effects on the current equilibrium price, of an
increase in the expected future amount of fruit from trees. The first is the
fact that at any given marginal utility of consumption of the fruit, the
higher expected fruit production increases the attractiveness of owning trees
today, which raises the current price of a tree. But, at the same time, the
increased expected fruit output in future periods means higher consumption and
lower marginal utility of consumption in that future period. This effect tends
to reduce the attractiveness of owning trees today - the tree is going to pay
off more in a time when marginal utility is expected to be low -- and thus
lowers the current value (and hence price) of a tree. These two forces are the
manifestation of the (pure) income effect and substitution effect from the
theory of \emph{consumer choice}, and their net result -- i.e which actual
dominates -- depends on the shape (and curvature) of the utility function.

In the case of logarithmic utility, income and substitution effects are
of the same size and opposite sign so the two forces exactly offset each
other, leaving the current price of a tree unchanged in the face of a rise in
expected future fruit output. This is (one of) the results from \citet{L1978}.
Although it does not appear the same powers are at force in
our stochastic mortality model, this does illustrate that there are a number
of settings in which one finds that 
logarithmic utility ($\gamma=1$) is the point of indifference between two
opposing consumption effects.

\section{Optimal Consumption: Numerical Examples}
\label{examples}

We started with a particular survival probability at time zero, namely the
Gompertz mortality curve with parameters $m=89.335$ and $b=9.5$. The age
$x=65$ survival probabilities to any age $y>x$ are given in Table 1. Both
hypothetical retirees agree on these numbers, which means that their initial
mortality rate is $\lambda_{0}=(1/9.5)\exp\{(65-89.335)/9.5\}=0.008125$.

Over time retiree 1 believes his mortality rate will grow at a rate
$\eta=(1/9.5)=\allowbreak0.105\,263\,16$ per year, while retiree 2 believes
it will evolve stochastically with a time-dependent growth rate of $\mu(t)$
and a volatility $\sigma.$ The actual curve $\mu(t)$ depends on the selected
parameter for volatility, since $\mu(t)$ is constrained to match
$p(0,\lambda_{0})$. The actual process for extracting $\mu(t)$ for any given
value of $\sigma$ is rather complicated (although it is not central to our
analysis) and is placed in the appendix of this paper. With these numbers in
hand -- and specifically the function $\mu(t)$ for the drift of the mortality
rate -- we can proceed to solve the PDEs given in equation (\ref{main.PDE.1})
and (\ref{main.PDE.2}), which then lead to the desired optimal consumption
function and the initial portfolio withdrawal rate at age 65.

\begin{figure*}
\begin{center}
\begin{tabular}
[c]{||c||c||c||c||c||c||c||}\hline\hline
\multicolumn{7}{||c||}{\textbf{Table 2: Optimal Retirement Portfolio
Withdrawal Rates $c^{\ast}(0)/F_{0}$}}\\\hline\hline
\textbf{Mortality Volatility} & $\gamma=0.5$ & $\gamma=1.0$ & $\gamma=1.5$ &
$\gamma=3$ & $\gamma=5$ & $\gamma=10$\\\hline\hline
$\sigma=0$ & 7.59\% & 6.12\% & 5.58\% & 5.02\% & 4.78\% & 4.61\%\\\hline\hline
$\sigma=15\%$ & 7.52\% & 6.12\% & 5.60\% & 5.04\% & 4.80\% &
4.62\%\\\hline\hline
$\sigma=25\%$ & 7.44\% & 6.12\% & 5.62\% & 5.06\% & 4.82\% &
4.63\%\\\hline\hline
\multicolumn{7}{||c||}{Notes: Retirement age 65, interest rate $r=2\%$,
mortality $\lambda_{0}=0.0081$}\\\hline\hline
\end{tabular}
\ \ 
\end{center}
\end{figure*}

Table 2 provides a variety of numerical examples across different values of
(mortality volatility) $\sigma$ and (risk aversion) $\gamma$, once again
assuming that the retirees are both at age $x=65$ with observable mortality
rate $\lambda_{0}=0.0081$. As we proved in Section \ref{reviewYaari}, and 
discussed above,
the consumption rate is the same across all levels of mortality volatility
when $\gamma=1$. It increases relative to DfM when $\gamma>1$ and decreases
relative to DfM when $\gamma<1$. Notice the impact of stochastic mortality on
optimal withdrawal rates is reduced as the value of risk aversion increases.
Notice how at a coefficient of relative risk aversion $\gamma=10$, the
portfolio withdrawal rates are approximately 4.6\% at all listed volatility levels.

Note that the $\sigma$ values provided are rather \emph{ad hoc} and have not
been estimated from any particular demographic dataset. We refer the
interested reader to recent actuarial papers -- such as \citet{BBRZ2008} --
for an empirical discussion around the estimation of the volatility of
mortality. Our objective here is to explore whether or not mortality volatility has a
(noticeable) impact on rational behavior as opposed to on insurance pricing.

To sum up, when the coefficient of CRRA (denoted by $\gamma$) is equal to one,
and the retiree has logarithmic utility preferences, the optimal consumption
rate at time zero is identical in both models. In other words, a retiree who
cannot adjust their consumption plan as mortality rates evolve starts-off with
the exact same consumption rate as the (more knowledgeable) consumer who can
adapt to changes in mortality rates and health status. Although the path of
their respective consumption will diverge over time -- depending on the
evolution of mortality rates -- initially they are the same. 
In contrast, when the coefficient of CRRA is greater than
one and the retiree is more risk-averse compared to a logarithmic utility
maximizer, the initial consumption rate is higher in the stochastic model vs.
the deterministic model. In other words, as one might expect the ability to
adapt to changes in health status and new information about mortality rates
allows the retiree to be more generous at time zero. Finally, when the
coefficient of the CRRA is between one and zero, the result is reversed. The
canonical retiree in a stochastic mortality model will consume less compared
to their neighbor who is operating under deterministic mortality assumptions.

Not withstanding the above results, the absolute consumption rate at time zero is
uniformly higher the lower the coefficient of relative risk aversion. This is
a manifestation of longevity risk aversion. The retiree is concerned about
living a long time, and therefore consumes less today to protect themselves
and self-insure consumption in old age.

\section{Discussion and Conclusion}
\label{conclusion}

In this article we extended the lifecycle model (LCM) of consumption over a
random-length lifecycle, to a model in which individuals can adapt behavior to
new information about mortality rates. The lifecycle model of saving and consumption continues to be very popular as a foundation model for decison-making amongst financial advisors,
as recently described in the monograph by \citet{BMS2008}.

\citet{Y1964,Y1965} was the first to include lifetime uncertainty in a
Ramsey-Modigliani lifecycle model and amongst other results, he provided a
rigorous foundation for Irving Fisher's claim that lifetime uncertainty
increases consumption impatience and is akin to higher subjective discount
rates. When the mortality rate itself is stochastic, this analogy is no longer
meaningful and -- to our knowledge -- the pure lifecycle model has not been
extended into the realm of 21st century models of mortality and longevity risk.

We built this extension by assuming that (i) the instantaneous force of
mortality is stochastic and obeys a diffusion process as opposed to being
deterministic, and (ii) that a utility-maximizing consumer can adapt their
consumption strategy to new information about their mortality rate (a.k.a.
current health status) as it becomes available. Our diffusion model for the
stochastic force of mortality was quite general, but inspired by (a.k.a.
borrowed from) the recent literature in actuarial science. We focused our
modeling attention on the retirement income\ stage of the LCM where health
considerations are likely to be more prevalent and to avoid complications
induced by wages, labor and human capital consideration.

In the first part of this paper we re-derived the optimal consumption function
under a deterministic force of mortality (DfM) using techniques from the
calculus of variations. We provided a closed-form expression for the entire
consumption rate function under a Gompertz mortality assumption. With those
benchmark results in place, we derived the optimal consumption strategy under
a stochastic force of mortality (SfM), by expressing and solving the relevant
Hamilton-Jacobi-Bellman (HJB) equation. In addition to the time variable, two
state variables in the resulting PDE are current wealth and the current
mortality rate. 

Retirees with (i) no bequest motives, (ii) constant relative risk aversion
(CRRA) preferences, and (iii) subjective discount rates equal to the interest
rate are expected to consume less as they age since they prefer to allocate
consumption into states of nature where they are most likely to be alive. This
is the conventional diminishing marginal utility argument. In our model, a
positive shock to the mortality rate in the form of pleasant health news
(perhaps a cure for cancer) will reduce consumption instantaneously and
further than expected at time zero. A negative shock to the mortality rate
(for example, being diagnosed with terminal cancer) will increase consumption
beyond what was expected.

Moving forward, a natural extension would be to explore the impact of
stochastic investment returns as well as mortality rates and include a
strategic asset allocation dimension, \emph{a la} \citet{M1971}. Another item
on our research agenda is to explore the optimal allocation to health and
mortality-contingent claims in a stochastic mortality model. Recall that one
of the noted results of \citet{Y1965} is that lifecycle consumers with no bequest
motives should hold all of their wealth in actuarial notes. However, in the
presence of a stochastic mortality, it is no longer clear how an insurance
company would price pension annuities, given the systematic risk involved. In
such a model, a retiree would have to choose between investing wealth in a
tontine pool, with corresponding stochastic returns or purchasing a pension
annuity with a deterministic consumption flow, but possibly paying a
risk-premium for the privilege. We conjecture that in a stochastic mortality
framework, the optimal \emph{product allocation} is a mixture of participating 
tontines and guaranteed annuities.

\bibliography{HMS}{}
\bibliographystyle{elsarticle-harv}

\appendix

\section{Matching Time-Zero Survival Curves}
\label{appendix}

The calibration of our economic model leads to an interesting by-product
problem in actuarial finance. In particular, in order to construct a
stochastic force of mortality that matches or fits a pre-determined Gompertz
survival curve -- the most popular and frequently used analytic law in this
literature -- one requires a lognormal diffusion process in which the drift
itself grows even faster than exponentially over time. In this appendix we
explain the mechanics of the procedure.

Given a deterministic model (Gompertz in our numerical examples), we compute
the time-zero survival function $p(t,\lambda_{0})$. We match this using a
stochastic model, by a suitable choice of parameters. This means that at time
0 the two models deliver identical survival probabilities. Recall that at
times other than $t=0$ the comparison will no longer be meaningful, even
controlling for the current observed mortality rate, because the mismatch
between conditional survival probabilities means that the two models give
different views of lifetimes going forward.

Let $\Lambda(t)=e^{-\int_{0}^{t}\lambda(q)\,dq}$ and define a pseudo-density
$q(t,\lambda)$ by the formula
\begin{equation}
E[\Lambda(t)\phi(\lambda(t))]=\int_{0}^{\infty}\phi(\lambda)q(t,\lambda
)\,d\lambda.
\end{equation}
Then $p(t,\lambda_{0})=\int_{0}^{\infty}q(t,\lambda)\,d\lambda$. By It\^o's
lemma,
\begin{multline}
\phi(\lambda(t))\Lambda(t)= \phi(\lambda_{0})+\int_{0}^{t}\Lambda(s)\Bigg[
\mu(s)\lambda(s)\phi^{\prime}(\lambda(s))\\
+\frac{\sigma^2}{2}\lambda
(s)^{2}\phi^{\prime\prime}(\lambda(s))-\lambda(s)\phi(\lambda(s))\Bigg]
\,ds\\
+\int_{0}^{t}\Lambda(s)\mu(s)\lambda(s)\phi^{\prime}(\lambda(s))\,dB(s).
\end{multline}
Take expectations and differentiate with respect to $t$. We get
\begin{multline}
\int_{0}^{\infty}\phi(\lambda)q_{t}(t,\lambda)\,d\lambda=\int_{0}^{\infty
}\Bigg[  \mu(t)\lambda\phi^{\prime}(\lambda)\\
+\frac{\sigma^{2}}{2}\lambda
^{2}\phi^{\prime\prime}(\lambda)-\lambda\phi(\lambda)\Bigg]  q(t,\lambda
)\,d\lambda
\end{multline}
with initial condition $q(0,\cdot)=\delta_{\lambda_{0}}$. Using integration by
parts (for $\phi$ vanishing fast at $0$ and $\infty$), we have
\begin{multline}
q_{t}(t,\lambda)
=-\mu(t)\frac{\partial}{\partial\lambda}\left[  \lambda
q(t,\lambda)\right]  \\
+\frac{\sigma^{2}}{2}\frac{\partial^{2}}{\partial
\lambda^{2}}\left[  \lambda^{2}q(t,\lambda)\right]  -\lambda q(t,\lambda).
\end{multline}
So if $\mu(t)$ is known for $0\leq t\leq t_{1}$, then all expectations
$\int_{0}^{\infty}q(t_{1},\lambda)\phi(\lambda)\,d\lambda$ can be found by
solving the forward equation for $q$ and then integrating against $\phi$.

Let $\lambda_{(1)}=\int_{0}^{\infty}\lambda qd\lambda$ and $\lambda_{(2)}%
=\int_{0}^{\infty}\lambda^{2}qd\lambda$ be the first two moments of
$q(t,\lambda)$. Note that the zeroth moment is the survival probability, so we
can integrate (by parts) the forward PDE for $q$ and the product of $\lambda$
and the forward PDE and obtain the following relationships
\begin{align}
\lambda_{(1)}  &  =-\frac{dp}{dt},\nonumber\\
\lambda_{(2)}  &  =\mu(t)\lambda_{(1)}-\frac{d\lambda_{(1)}}{dt}.
\end{align}
Combined the two expressions, we have
\begin{equation}
\mu(t)=\frac{\frac{d\lambda_{(1)}}{dt}+\lambda_{(2)}}{\lambda_{(1)}}.
\end{equation}
Replacing $\mu(t)$ in the forward PDE for $q$ and obtain an
integro-differential equation
\begin{multline}
q_{t}(t,\lambda)=-\frac{\frac{d\lambda_{(1)}}{dt}+\lambda_{(2)}}{\lambda
_{(1)}}\frac{\partial}{\partial\lambda}\left[  \lambda q(t,\lambda)\right]\\
+\frac{\sigma^{2}}{2}\frac{\partial^{2}}{\partial\lambda^{2}}\left[
\lambda^{2}q(t,\lambda)\right]  -\lambda q(t,\lambda),
\end{multline}
or
\begin{multline}
q_{t}(t,\lambda)=\frac{-\frac{\partial^{2}p}{\partial t^{2}}+\int_{0}^{\infty
}\lambda^{2}q(t,\lambda)d\lambda}{\frac{\partial p}{\partial t}}\frac
{\partial}{\partial\lambda}\left[  \lambda q(t,\lambda)\right]  \\
+\frac{\sigma^{2}}{2}\frac{\partial^{2}}{\partial\lambda^{2}}\left[  \lambda
^{2}q(t,\lambda)\right]  -\lambda q(t,\lambda),
\end{multline}
which we can solve numerically with the initial condition $q(0,\lambda
)=\delta(\lambda-\lambda_{0})$.

We solve the integro-differential equation for $q$ numerically first, obtain
the value of $\mu(t)$. We then solve the HJB equation for optimal consumption
as before, with the constant $\mu$ now replaced by the function $\mu(t)$ at
$\lambda=0$.

Finally, we should record a couple of remarks about the form of $\mu(t)$.
First of all,
\begin{equation}
\mu(0)=\eta.
\end{equation}
To see this, observe that $p_{t}(0,\lambda_{0})=-E[\lambda_{0}\Lambda
_{0}]=-\lambda_{0}$. So $\lambda_{0}^{2}=E[\lambda_{0}^{2}\Lambda_{0}%
^{2}]=p_{tt}(0,\lambda_{0})+\lambda_{0}\mu(0)$. But $p_{tt}(0,\lambda_{0})$
can be computed explicitly, since it is Gompertz, to give $\lambda_{0}%
^{2}-\lambda_{0}\eta$. This implies that $\mu(0)=\eta$.

Second, note that $\mu(t)$ should be increasing in $\sigma$. The mean
$E[e^{-\int_{0}^{t} \lambda(q)\,dq}]$ does not change with $\sigma$, so by
convexity of the exponential, the median of this quantity must decrease as we
increase the variance. In other words, $\mu(t)$ must rise. Put another way,
this expectation is driven by the possibility of relatively larger values of
the exponent, ie of abnormally low values of $\lambda$. As $\sigma$ rises, the
impact of longevity risk gets more pronounced, and to compensate for that the
growth rate $\mu(t)$ must also rise.

\end{document}